\documentclass[pra, twocolumn, superscriptaddress, a4paper, article]{revtex4-1}

\usepackage{amsmath,amssymb,amsthm}
\usepackage{quantum}
\usepackage{color}
\usepackage{subcaption}
\usepackage{tikz}
\usepackage{hyperref}

\newtheorem{lemma}{Lemma} 
\newtheorem{corollary}{Corollary}

\newtheorem{theorem}{Theorem}

\begin{document}
\title{Spectral properties for a family of two-dimensional quantum antiferromagnets}

\author{Andrew S. Darmawan}
\affiliation{D\'epartement de Physique, Universit\'e de Sherbrooke, Qu\'ebec, Canada}
\affiliation{Centre for Engineered Quantum Systems, School of Physics, The University of Sydney, Sydney, NSW 2006, Australia}
\author{Stephen D. Bartlett}
\affiliation{Centre for Engineered Quantum Systems, School of Physics, The University of Sydney, Sydney, NSW 2006, Australia}
\date{21 January 2016}

\begin{abstract}
    We study the spectral properties of a family of quantum antiferromagnets on two-dimensional (2D) lattices.   This family of models is obtained by a deformation of the well-studied 2D quantum antiferromagnetic model of Affleck, Kennedy, Lieb and Tasaki (AKLT); they are described by two-body, frustration-free Hamiltonians on a three-colourable lattice of spins.  Although the existence of a spectral gap in the 2D AKLT model remains an open question, we rigorously prove the existence of a gap for a subset of this family of quantum antiferromagnets.  Along with providing new progress for the gap problem in AKLT-type antiferromagnets in 2D, this result has implications for the theory of quantum computation as it provides a family of two-body Hamiltonians for which the ground state is a resource for universal quantum computation and for which a spectral gap is proven to exist.
\end{abstract}
\maketitle

\section{Introduction}

Quantum antiferromagnetic spin lattice models possess a rich structure, and exhibit many properties of quantum phases that lie beyond the standard Landau classification by symmetry breaking.  In one-dimensional chains, quantum antiferromagnetic models were the subject of  Haldane's conjecture: that there is a gap in the energy spectrum of integer-spin antiferromagnetic chains whereas half-odd-integer-spin chains are gapless.  The analytically-solvable model of Affleck, Kennedy, Lieb and Tasaki (AKLT) for a one-dimensional spin-$1$ chain possesses a spectral gap, a ground state degeneracy determined by the boundary conditions, and a nonlocal `string-like' order parameter \cite{affleck_rigorous_1987, affleck_valence_1988, kennedy_hidden_1992-1}; this model now serves as the canonical example of a one-dimensional spin chain possessing both symmetry-protected topological order~\cite{GuWen2009} and a valence bond solid description.  For higher-dimensional lattices, quantum antiferromagnetic models are expected to possess a similarly rich structure, but classifying these structures is a challenge, including even the rigorous identification of a spectral gap in natural generalisations of the AKLT state to 2D lattices.

Methods from the theory of quantum computation provide a new approach to identifying and classifying the quantum order present in antiferromagnetic spin lattice models~\cite{VerstraeteCirac2004}.  In particular, it has been shown for 1D antiferromagnetic spin chains that their unique type of quantum order can be identified with an infinite localizable entanglement length~\cite{VerstraeteMDCirac2004,CamposVenuti2005,brennen_measurement-based_2008, else_symmetry-protected_2012}.  For the 2D AKLT model and variants on a number of spin-$3/2$ lattices and the spin-2 square lattice, an even stronger result has been proven:  that the ground state of such models possesses an entanglement structure that allows for universal quantum computation through local measurements alone \cite{wei_affleck-kennedy-lieb-tasaki_2011, miyake_quantum_2011, wei_quantum_2013-1, wei_universal_2015-1}. Quantum computation that proceeds by single-particle measurements of a many-body entangled quantum state is called measurement-based quantum computation (MBQC). The universality of the AKLT model for measurement-based quantum computation was proven by demonstrating that its ground states, on certain lattices, can be converted through local operations into a class of entangled states known as \emph{graph states}---states that form the central resource for the theory of measurement-based quantum computation~\cite{briegel_persistent_2001, raussendorf_one-way_2001}.  From the perspective of quantum computing, there is keen interest in identifying and classifying the types of quantum order in such spin lattices that enable quantum computation in this sense \cite{bartlett_simple_2006, griffin_spin_2008, brennen_measurement-based_2008, chen_gapped_2009, cai_universal_2010, wei_affleck-kennedy-lieb-tasaki_2011, miyake_quantum_2011, li_thermal_2011,wei_universal_2015-1, darmawan_measurement-based_2012, else_symmetry_2012-1, fujii_measurement-based_2013}. 

In this paper, we show that methods from quantum information theory, in particular the study of graph states and AKLT states for quantum computation, can be used in proving spectral properties of 2D antiferromagnetic spin models.  In particular, we consider a family of 2D antiferromagnetic spin models obtained by a one-parameter deformation of the 2D AKLT model.  This model was introduced in a previous work~\cite{darmawan_graph_2014-1}, wherein it was shown using numerical methods that there is a region in parameter space where the ground state of such models is universal for quantum computation.  The deformation to the AKLT model we consider is a modification of that of Ref.~\cite{niggemann_quantum_1997}. In contrast to our case, however, the ground state of the model in Ref.~\cite{niggemann_quantum_1997} tends to a N{\'e}el state (rather than a graph state) as the deformation becomes large.  

Here, we prove several rigorous results concerning the spectral properties of our family of antiferromagnets, with a focus on the key question of whether there exists a spectral gap separating the ground state from the first excited state in the thermodynamic limit.  While the AKLT model on which the model is based is not currently known to be gapped, we prove that a subset of this family of models (not including the AKLT point) is gapped provided the ground state is sufficiently close to a graph state.  Our work complements other studies that have numerically (rather than analytically) investigated the gap in other families of antiferromagnets based on the AKLT model~\cite{Ganesh,GarciaSaez}.

One implication of our results is that this family of models has many of the desired properties for quantum computation, in particular that the Hamiltonian is two-body and that there exists a finite region in the parameter space that is both computationally universal and gapped.  Previous works have presented 2D models, based on the AKLT model, that are both gapped and universal for MBQC~\cite{cai_universal_2010, li_thermal_2011}.  However, these models do not possess the full 2D structure of the family of models presented here. The model in~\cite{cai_universal_2010} is unitarily equivalent to a product of 1D chains, and the model in~\cite{li_thermal_2011} can be transformed via a unitary into decoupled regions of finite size.

The paper is structured as follows.  In Sec.~\ref{s:models}, we define the quantum spin lattice models that we investigate:  the AKLT model, the graph state model, and our family of 2D quantum antiferromagnets.  We then investigate the spectral properties of these models in Sec.~\ref{s:spectral}, including a rigorous proof of the existence of a spectral gap for a subset of these 2D quantum antiferromagnets.  We summarise our results in Sec.~\ref{s:conclusion}.  Several technical results are relegated to appendixes.

\section{AKLT models, graph state models, and a family of 2D antiferromagnets}
\label{s:models}

In this section we define AKLT states and graph states in terms of their parent Hamiltonians.  We note that AKLT states and graph states are special in that both are ground states of frustration free Hamiltonians (i.e., both AKLT states and graph states are the simultaneous ground states of all individual terms in their respective Hamiltonians).  In addition, they are describable as projected entangled pair states (PEPS)~\cite{verstraete_matrix_2008} and can be created from product states using constant-depth unitary circuits. 

We then define a class of 2D antiferromagnet models obtained as deformed versions of the AKLT model, generalising the spin-$3/2$ models originally presented in Ref.~\cite{darmawan_graph_2014-1}.  These deformed AKLT models have ground states that interpolate between an AKLT state and a graph state.

\subsection{Parent Hamiltonians for AKLT and graph states}
\label{s:parent_hams}

Here, we review parent Hamiltonians for graph states and AKLT states.  Alternative equivalent definitions of these states in terms of tensor network states are detailed in Appendix~\ref{s:peps_defs}.

A parent Hamiltonian for a state $\ket{\psi}$ of $N$ particles is a local Hamiltonian $H=\sum_{i}h_i$ that has $\ket{\psi}$ as its unique lowest energy eigenstate. By local, we mean that each interaction term $h_i$ acts non-trivially on at most $k$ particles, where $k$ is a constant independent of $N$.  The Hamiltonians that we consider in this paper are also frustration-free, meaning that the ground state of $H$ is the lowest eigenstate of each interaction term $h_i$ individually.  The minimum energy of each $h_i$ term may be set to zero by adding a multiple of the identity, and by replacing each $h_i$ with the projector onto its image we will not change the ground space of the Hamiltonian nor any crucial spectral properties such as the existence of a spectral gap.  Therefore, for simplicity, in this paper we need only consider frustration-free Hamiltonians that are sums of projectors.  

\subsubsection{AKLT states}

The standard AKLT parent Hamiltonian may be defined on an arbitrary graph $G=(V,E)$ as follows. At each vertex $v$ of the graph place a spin-$S_v$ particle, where $S_v=(d_v-1)/2$ and $d_v$ is the degree of vertex $v\in V$. The AKLT Hamiltonian is an antiferromagnetic interaction Hamiltonian given by
\begin{equation}
    H^A:=\sum_{\langle i,j \rangle} P_{ij}^{S_{\rm tot}=S_i+S_j}\,,
    \label{e:akltham}
\end{equation}
where $S_i$ is the spin of particle $i$, and $P^{S_{\rm tot}=S}_{ij}$ is the projection onto the spin-$S$ irreducible representation of particles $i$ and $j$. Defined in this way, $H^A$ has a unique zero-energy eigenstate, which we call the AKLT state $\ket{\rm AKLT}$. Furthermore, the AKLT state is a zero eigenstate of each term in \eqref{e:akltham}, implying that the Hamiltonian is frustration-free. We note that the terms in this Hamiltonian do not commute with each other, and so establishing the existence of a spectral gap for such models is highly non-trivial.  For two-dimensional lattices, it remains an open question in general whether this AKLT Hamiltonian possesses a spectral gap in the thermodynamic limit. 

\subsubsection{Graph states}

A graph state, and its corresponding parent Hamiltonian, may also be defined on an arbitrary graph $G=(V,E)$. In contrast to AKLT states, however, every particle in a graph state is taken to have spin-1/2. (While we do not consider it here, it is possible to generalise to higher spin~\cite{bahramgiri_graph_2006}.) Following the original definitions \cite{briegel_persistent_2001, raussendorf_one-way_2001}, the graph state $\ket{G}$ corresponding to a graph $G$ is defined as the output of a simple quantum circuit
\begin{equation}
    \ket{G}=\left(\prod_{e\in E } {\rm CZ}_e\right)\left(\bigotimes_{v\in V } \ket{+}_v\right)\,,
    \label{e:gs_definition}
\end{equation}
i.e., a qubit in the state $\ket{+}:=\rt{2}(\ket{0}+\ket{1})$ is placed at each vertex, and a controlled-$Z$ gate ${\rm CZ}:=\exp(i\pi\ketbra{11}{11})$ is applied between every pair of vertices connected by an edge. 

In contrast to AKLT states, graph states possess a particularly simple parent Hamiltonian for which the spectral properties are much simpler.  The graph state is the unique ground state of the Hamiltonian 
\begin{equation}
    H^C:=\sum_{i=1}^N \fr{2}(I-K_i)\,,
    \label{e:cluster_ham}
\end{equation}
where $K_i:=X_i \prod_{j\in n_i} Z_j$ are stabilizer generators, $X_i$ and $Z_i$ are Pauli $X$ and $Z$ matrices respectively acting on particle $i$, and where the product is taken over all neighbours $n_i$ of particle $i$. Note that the interaction terms $\{\fr{2}(I-K_i): i=1,\dots, N\} $ pairwise commute. So, unlike the AKLT Hamiltonian, this cluster Hamiltonian is trivially diagonalisable and possesses a unit spectral gap.

\subsubsection{Relationship between AKLT and graph states}

Motivated by physical considerations, one generally views the Hamiltonian $H^A$ as being more realistic than $H^C$, as it involves only two-body interaction terms.  The cluster Hamiltonian $H^C$ involves many-body interactions.  Even for a one-dimensional chain (degree-two), the interaction terms $K_i$ of $H^C$ each act non-trivially on three particles, and the bodiness of the interaction terms increases with the degree of the graph.  Despite these somewhat artificial features, the cluster Hamiltonian $H^C$ has commuting terms and a unit gap, and its simplicity makes it a useful tool for investigating spectral properties of more general models.  

Let $G$ be three-colourable graph and let $\ket{\rm AKLT}$ be the AKLT state defined on this graph.  AKLT states on three-colourable graphs have the desirable property that they can be transformed into graph states on the same graph using a set of single-particle projectors \cite{wei_affleck-kennedy-lieb-tasaki_2011}. We use this property to construct a continuous family of Hamiltonians and their associated ground states that interpolate between these AKLT states and the corresponding graph states.

Define three operators 
\begin{equation}
    P^c:=\ket{S}_c\bra{S}+\ket{-S}_c\bra{-S}\,,
    \label{e:projector_def}
\end{equation}
for $c\in\{x,y,z\}$ where $\ket{{\pm}S}_c$ are spin-$S$ states satisfying $\hat{S}^c\ket{{\pm}S}_c={\pm}S\ket{{\pm} S}_c$ and $\hat{S}^c$ is the single-spin operator for a spin-$S$ particle along the $c$ axis. In other words, the $P^c$ operators are projections onto the two-dimensional subspace spanned by $\ket{{\pm} S}_c$.

Now, let $\{c_i\in \{x,y,z\}:i=1\dots N\}$ be a three-colouring of $G$, i.e., $c_j\ne c_i$ if $i$ and $j$ are neighbours. For each particle $i$ in our system, we define $P^{c_i}_i$ with respect to Eq. \eqref{e:projector_def}, with $S$ taken to be the spin of particle $i$. We will call the image of $P^{c_i}_i$ the logical space $\mathcal{H}_L^i$ of particle $i$ and we will associate the physical spin states $\ket{S}_c$ and $\ket{-S}_c$ with the logical qubit states $\ket{0}$ and $\ket{1}$ respectively. Then the projectors $P^{c_i}_i$ map the AKLT state to a graph state in the sense that
\begin{equation}
    \bigotimes_{i=1}^N P^{c_i}_i\ket{{\rm AKLT}}\propto\ket{G}\,,
\end{equation}
where $\ket{G}$ is a graph state on the graph $G$, where qubit $i$ is encoded into the logical subspace $\mathcal{H}_L^i$ as above. This result was originally derived in \cite{wei_affleck-kennedy-lieb-tasaki_2011}, and we provide an alternative proof using tensor networks in Appendix~\ref{s:local_conversion}. 

\subsection{A family of two-body antiferromagnets}
\label{s:deform_definition}

In this section we will construct a family of deformed AKLT-like antiferromagnet Hamiltonians, originally described in Ref.~\cite{darmawan_graph_2014-1} for the special case of spin-3/2 models, that interpolate between the physically-motivated but more complicated AKLT Hamiltonian $H^A$ and a simpler model with a structure based on graph states.  This Hamiltonian has the appealing property that it has only two-body interactions. 

For each site $i$, we define a one-parameter family of single-particle operators $D_i(\delta)$ acting on site $i$ as
\begin{equation}
    \label{e:def_deformations}
  	D_i(\delta):=(1-\delta)P^{c_i}_i+\delta I \,,
\end{equation}
where $0 \leq \delta\leq 1$. This operator satisfies $D_i(1)=I$ and $D_i(0)=P^{c_i}_i$. Hence the (unnormalised) state
\begin{equation}
    \ket{\psi(\delta)}:=\bigotimes_{i=1}^N D_i(\delta)\ket{\rm AKLT}\,,
    \label{e:state_definition}
\end{equation}
interpolates smoothly between $\ket{\rm AKLT}$ at $\delta=1$ and $\ket{G}$ at $\delta=0$. If we restrict to the domain $\delta>0$, then $D_i(\delta)$ is invertible. Given that $\ket{\rm AKLT}$ is the unique state annihilated by every $P^{S_{\rm tot}=S_i+S_j}_{ij}$, it follows that $\ket{\psi(\delta)}$ is the unique state annihilated by every deformed interaction term
\begin{equation}
    \left[D_i(\delta)\iv \otimes D_j(\delta)\iv\right]P^{S_{\rm tot}=S_i+S_j}_{ij}\left[D_i(\delta)\iv \otimes D_j(\delta)\iv\right]\,.
    \label{e:deformed_interaction}
\end{equation}
We remark that these couplings take the form $ABA$ where $A$ and $B$ are positive operators, thus the couplings themselves are positive. (It is not a simple unitary change of basis, which would leave the spectrum of the Hamiltonian unchanged.)  Note that these couplings have an unbounded spectrum in the region $\delta>0$ with some eigenvalues tending to infinity as $\delta\rightarrow 0$. 

To ensure that the spectral properties of the family of models is not obscured by the diverging spectral behavior of the local couplings, we replace the interaction terms with a set of projectors $\{Q(\delta)_{ij}: (i,j)\in E\}$ where $Q(\delta)_{ij}$ is defined as the projection onto the image of \eqref{e:deformed_interaction}. 

Thus we define a family of two-body Hamiltonians
\begin{equation}
    H(\delta)=\sum_{\langle i,j\rangle} Q(\delta)_{ij}\,,
    \label{e:twobody_gs_ham}
\end{equation}
which has $\ket{\psi(\delta)}$ as a unique ground state. Given that the operators $Q(\delta)_{ij}$ are well-defined for any $\delta>0$, and that $\lim_{\delta\rightarrow 0 } \ket{\psi(\delta)}=\ket{G}$, we conclude that the ground state of $H(\delta)$ can be made arbitrarily close (in terms of fidelity) to $\ket{G}$ by setting $\delta>0$ sufficiently small. 

The interaction terms $Q(\delta)_{ij}$ are defined for any $\delta>0$, however at $\delta=0$ the operators $D(\delta)$ become singular and thus Eq. \eqref{e:deformed_interaction} is not defined. We may nevertheless define the Hamiltonian at $\delta=0$ by taking the limit $Q(0):=\lim_{\delta\rightarrow 0}Q(\delta)$ and letting $H(0):=\sum_{\langle i,j\rangle} Q(0)_{ij}$. This limit exists, and can in principle be computed by applying Gram-Schmidt orthogonalisation to the column span of Eq \eqref{e:deformed_interaction} for non-zero $\delta$, then taking the limit as $\delta\rightarrow0$. $H(\delta)$ is continuous and $Q(\delta)$ have the same rank in the entire interval $\delta\in[0,1]$, however we note that the ground state degeneracy of $H(\delta)$ increases drastically at $\delta=0$. This behaviour is studied in greater detail in the following section. This essentially implies that the ground state of $H(\delta)$ cannot be made exactly equal to a graph state by setting $\delta=0$, despite the fact that it can be made arbitrarily close by setting $\delta$ small and non-zero.

\section{Spectral properties}
\label{s:spectral}

The construction detailed in the previous section provides us with a family of two-body, frustration-free antiferromagnetic Hamiltonians $H(\delta)$ of Eq.~\eqref{e:twobody_gs_ham} that interpolates between a Hamiltonian possessing an AKLT state as its unique ground state for $\delta=1$ and a Hamiltonian possessing a unique ground state that is arbitrarily close to a graph state for $\delta \rightarrow 0$. Like the AKLT Hamiltonian $H^A$, neighbouring interaction terms $Q(\delta)_{ij}$ do not commute, and analysing spectral properties of $H(\delta)$ is non-trivial. In this section, we study spectral properties of the Hamiltonian $H(\delta)$ of Eq.~\eqref{e:twobody_gs_ham} by making use of the spectral properties of the simpler Hamiltonian $H^C$ of Eq.~\eqref{e:cluster_ham}.

In particular, we will study the behaviour of the spectral gap separating the ground state energy and the first excited state. To be precise, for fixed integer $N$, we will define the gap $\Delta_N=E_1-E_0$ to be the difference between the smallest and second smallest eigenvalues of the Hamiltonian with system size $N$. Note that $E_0=0$ due to our Hamiltonian being a frustration-free sum of projectors. We say that a Hamiltonian is gapped if there exists a constant $\Delta>0$ and integer $M$ such that $\Delta_N>\Delta$ for all $N>M$.  

Proving that a model is gapped is hard in general \cite{cubitt_undecidability_2015}. Whether or not the AKLT Hamiltonain $H^A$ is gapped on two-dimensional lattices remains an open question. The family $H(\delta)$ of Hamiltonians contains the AKLT Hamiltonian $H(1)=H^A$, and thus proving that $H(\delta)$ is gapped for all $\delta>0$ is at least as hard as proving a gap for $H^A$.  We use a range of methods to study the spectral gap in the region where $\delta$ is small and the ground state is close to a graph state. 

Note that there are two key parameters determining the size of the gap of $H(\delta)$: the parameter $\delta$ (which reflects the closeness of the ground state to a graph state), and the system size $N$. First we will consider the case where $N$ is fixed and $\delta$ is allowed to vary. We will show that the gap, as a function of $\delta$, approaches zero as $\delta$ approaches zero. In other words, there is a tradeoff between the fidelity of the ground state with a graph state, and the gap of the Hamiltonian. 

While small $\delta$ implies small gap for finite-sized systems, we will then show in Sec.~\ref{s:thermodynamic_limit} that small $\delta$ also guarantees the existence of a gap in the thermodynamic (infinite system size) limit.  In other words, we prove that for $\delta$ fixed and sufficiently small, there exists a constant $\Delta$ independent of system size, for which $\Delta_N > \Delta$ for all $N$.

\subsection{Behaviour of the gap as $\delta \rightarrow 0$ for finite  lattices}
\label{s:finite_gap}

Here we will prove that, for fixed system size, there is a tradeoff between the gap of the Hamiltonian and the fidelity of the ground state with a graph state. We stated this result previously in \cite{darmawan_graph_2014-1}, although did not provide a rigorous proof. 

The Hamiltonians we consider are not related simply to the stabilizer Hamiltonian for a graph state (which is clearly gapped).  In fact, there is no point in the family where the ground state is exactly a graph state.  We show in this section that for any finite sized system the gap shrinks to zero in the limit where the fidelity of the ground state with a graph state tends to one. This result necessarily comes about because we have restricted to two-body Hamiltonians for which certain ground states (e.g., graph states) cannot be reached without closing the gap.

We will first prove a result about a class of graph states which we call \emph{non-trivial} graph states. A graph state is non-trivial if the graph $G$ on which it is defined satisfies the following two properties. Firstly, the degree of each vertex of $G$ is at least two. Secondly, no two distinct vertices in $G$ have an identical set of neighbours. These properties are satisfied on all 2D lattices considered in this paper. We have the following lemma:

\begin{lemma}
    Let $\ket{G}$ be a non-trivial graph state on a graph of spin-$1/2$ particles. Let $A:=\{i,j\}$ be any two sites and $B$ be its complement. Then the reduced density operator $\rho_A={\rm tr}_B (\ketbra{G}{G})$ is $\fr{4}I$. Moreover, the only positive, two-body, spin-$1/2$ Hamiltonian with $\ket{G}$ as a zero energy eigenstate is the zero Hamiltonian.
    \label{t:trivial_gs_ham}
\end{lemma}
\begin{proof}
We will use the fact that the reduced density operator $\rho_A$ of a stabilizer state with stabilizer $S$ is given by 
\begin{equation}
    \rho_A=\fr{2^{|A|}}\sum_{\sigma\in S_A} \sigma\,,
    \label{e:reduced}
\end{equation}
where $S_A\subseteq S$ is the set of stabilizer elements $\sigma\in S$ which act as the identity outside $A$. Eq. \eqref{e:reduced} will hold for any set of particles $A$, however here we will set $A$ to be an arbitrary two-particle region $\{i,j\}$. 

Any element in the stabilizer $\sigma\in S$ can be expressed uniquely as $\sigma=\prod_{l=1}^N K_l^{\gamma_l}$, where $K_l$ are the graph state stabilizer operators defined in Eq. \eqref{e:cluster_ham} and $\gamma_l\in \{0,1\}$ for each $l$. It is clear that any $\sigma\in S$ with $\gamma_l=1$ for $l\notin A$ cannot be an element of $S_A$. This eliminates all but four possible stabilizer elements from membership in $S_A$: $K_i$, $K_j$, $K_i K_j$ and $I$. As we have assumed that the degree of $G$ is at least 2, $K_i$ and $K_j$ must act non-trivially outside $A$, and thus $K_i$ and $K_j$ are not contained in $S_A$. Furthermore, from the assumption that $i$ and $j$ do not have an identical set of neighbours outside $A$, it is clear that $K_i K_j$  is also not contained in $S_A$. Therefore $S_A=\{I\}$ and from Eq. \eqref{e:reduced}, $\rho_A=\fr{4}I$. Thus any positive non-zero interaction term acting on $A$ will have positive, non-zero energy. Since this holds for every pair of particles in the state, the only two-body Hamiltonian with positive interaction terms and zero energy will be the zero Hamiltonian. 
\end{proof}

The above result holds for spin-$1/2$ Hamiltonians, however it does not directly apply to our family of Hamiltonians, which generally involve higher-dimensional particles. We will now generalise to the case where each graph state qubit is encoded in a two-dimensional subspace, called the logical space $\mathcal{H}_L^i$, of a higher-dimensional physical particle. We find that any two-body, frustration-free Hamiltonian with such a graph state in its ground space must act trivially on the logical space of each particle. 
\begin{lemma}
    \label{t:qudit_gs_ham}
    Consider a system of $N$ particles $\otimes_{i=1}^N\mathcal{H}^i$, where the dimension $d_i=2j+1$ of each spin-$j$ system $\mathcal{H}^i$ is at least 2. Let $\ket{G}$ be a non-trivial graph state with $N$ vertices defined such that the $i$-th qubit is encoded into a two-dimensional subspace of the $i$-th physical particle $\mathcal{H}^i_L\subseteq \mathcal{H}^i $. Let $h_{ij}$ be a positive, two-body operator acting on $\mathcal{H}^i\otimes \mathcal{H}^j$ and satisfying $h_{ij}\ket{G}=0$. Then $h_{ij}$ is supported entirely on $(\mathcal{H}_L^i \otimes \mathcal{H}_L^j)^\perp$, i.e. $h_{ij}=P_L^\perp h_{ij} P_L^\perp$ where $P_L^\perp$ is the projector onto $(\mathcal{H}_L^i \otimes \mathcal{H}_L^j)^\perp$. 
\end{lemma}
\begin{proof}
    Let $A=\{i,j\}$ be a two particle region. Let $P_L$ be the projection onto $\mathcal{H}_L^i\otimes \mathcal{H}_L^j$. 
Using a decomposition of the identity $I=P_L+P_L^\perp$, we have 
\begin{equation}
    h_{ij}=P_Lh_{ij}P_L+(P_L^\perp h_{ij}P_L + {\rm h.c.}) + P_L^\perp h_{ij}P_L^\perp\,.
    \label{e:decompose}
\end{equation}
From $h_{ij}\ket{G}=0$ we have that $P_L h_{ij}\ket{G}=0$ and $P_L^{\perp}h_{ij}\ket{G}=0$. From Lemma \ref{t:trivial_gs_ham}, we also have that the reduced density operator on the region $A$ satisfies $\rho_A=\fr{4}P_L$. By rearranging ${\rm tr}(P_L h_{ij}\ketbra{G}{G} h_{ij} P_L)=0$ and ${\rm tr}(P_L^\perp h_{ij}\ketbra{G}{G} h_{ij} P_L^\perp)=0$, we obtain $||P_L h_{ij} P_L||_2=0$ and $||P_L^\perp h_{ij} P_L||_2=0$ where $||B||_2:=\sqrt{{\rm tr}(B,B)}$ is the Frobenius norm. Hence only the last term in Eq. \eqref{e:decompose} is non-zero.
\end{proof}

This implies that any frustration-free, two body Hamiltonian acting on qudits with a graph state as an exact ground state has the undesirable property that the ground space contains the entire logical space $\otimes_{i=1}^N\mathcal{H}_L^i$. Thus the ground space of the Hamiltonian is (at least) $2^N$-fold degenerate. From this it follows trivially that our Hamiltonian $H(\delta)$ defined in Eq. \eqref{e:twobody_gs_ham} acting on a system of fixed size $N$ with a non-trival graph state $\ket{G}$ as an approximate ground state, has a gap that shrinks to zero as $\delta$ tends to zero. We prove the result generically for any two-body, frustration-free Hamiltonian with a approximate graph state as a unique ground state.

\begin{theorem}
    Let $H(\delta)=\sum_{\langle i,j\rangle}h_{ij}(\delta)$ be a two-body, frustration-free Hamiltonian acting on $N$ qudits, defined on the interval $\delta\in[0,1]$, where each term $h_{ij}(\delta)$ is a projector that varies continuously in $\delta$. Assume that for $\delta>0$, $H(\delta)$ has a unique ground state $\ket{\psi(\delta)}$ which has the property that $\lim_{\delta\rightarrow 0}\ket{\psi(\delta)}=\ket{G}$, where $\ket{G}$ is a non-trivial graph state with $N$ vertices (in the sense of Lemma \ref{t:qudit_gs_ham}). Then $\lim_{\delta\rightarrow0}\Delta_N=0$.
\end{theorem}
\begin{proof} 
    By continuity, each $h(0)_{ij}$ has $\ket{G}$ as a zero eigenstate. By Lemma \ref{t:qudit_gs_ham}, $H(0)$ contains the $2^N$-dimensional logical space $\otimes_{i=1}^N\mathcal{H}_L^i$ in its ground space. The fact that $H(\delta)$ has a unique ground state while $H(0)$ has a groundspace that is at least $2^N$-dimensional, and the fact that $H(\delta)\rightarrow H(0)$ as $\delta\rightarrow 0 $, implies a large number of zero-energy crossings at $\delta=0$. The existence of ground state energy crossings implies $\lim_{\delta\rightarrow 0}\Delta_N=0$. 
\end{proof}

We have thus shown for any finite-sized system, that if $H(\delta)$ is a two-body, frustration-free Hamiltonian with a ground state that tends to a graph state as $\delta$ tends to zero, then the gap of the Hamiltonian must also tend to zero as $\delta$ tends to zero. We remark that if the bodiness of the Hamiltonian is greater than two, this trade-off is not necessarily observed. For instance, it was shown in Ref.~\cite{darmawan_graph_2014-1} that for three-body parent Hamiltonians in one dimension, the fidelity can be improved arbitrarily without decreasing the size of the gap.

\subsection{Gap in the thermodynamic limit}
\label{s:thermodynamic_limit}

We now prove that the Hamiltonian $H(\delta)$ is gapped in the thermodynamic limit if $\delta>0$ is independent of $N$ and sufficiently small. This statement may seem at odds with the result of the previous section, where we showed that the gap of the finite-sized system tends to zero as $\delta$ goes to zero. The two results can be reconciled by the fact that we are taking two different limits: here we are considering the case where $\delta$ is fixed and $N$ tends to infinity, while in the previous section we considered the case where $N$ was fixed and $\delta$ tends to zero.

Our proof relies on properties of parent Hamiltonians for injective tensor network states.  Injectivity is a generic property of tensor network states on a lattice~\cite{perez-garcia_peps_2008}, and is useful in our proof as it implies the existence of a unique ground state. To prove injectivity, we first perform a blocking and construct a blocked parent Hamiltonian, for which we can prove a spectral gap that is stable for small perturbations in $\delta$ around zero.  After defining tensor network states and blocked parent Hamiltonians, we then prove that the class of two-body antiferromagnetic Hamiltonians $H(\delta)$ is gapped on various lattices for sufficiently small $\delta$ and provide lower bounds on the size of the gapped region for trivalent lattices (i.e., the largest $\delta$ for which we can prove $H(\delta)$ is gapped).

\subsubsection{Tensor network states and PEPS}

The projected entangled-pair state (PEPS) framework provides an efficient description of a class of multipartite quantum states in arbitrary spatial dimensions. Here we provide a brief description of the PEPS framework, as it applies to our problem; for further details, see Ref.~\cite{verstraete_matrix_2008}. 

A tensor is a multi-index array of complex numbers $A_{\alpha_1, \alpha_2, \dots, \alpha_n}$ where each index $\alpha_k$ has some finite dimension $d_k$, which we call the bond dimension. Graphically, we represent a tensor as in Fig.~\ref{f:blocking}(a).  In this graphical framework, when two tensors are connected by an edge, their corresponding indices are contracted.  For the purposes of defining tensor network states, we identify two distinct types of indices.  The first type are physical indices, denoted $\beta_i$, associated with an orthonormal basis $\{ |\beta_i\rangle, i = 1,\ldots d \}$ for each elementary $d$-dimensional quantum system.  The second type are virtual indices, denoted using $\alpha$ indices, which are not associated with physical degrees of freedom but are contracted internally within the tensor network to allow for a non-product entanglement structure.  For clarity, in our tensor network diagrams we will not usually draw an edge for each physical index, instead implying the existence of a physical index on a tensor with a superscript $\beta$ on the tensor label. 

Consider a quantum state defined in the above basis as $\ket{\psi}=\sum_{\beta_1,\beta_2,\dots,\beta_n}\psi^{\beta_1,\beta_2,\dots,\beta_n}\ket{\beta_1,\beta_2,\dots,\beta_n}$.  A projected entangled pair state (PEPS) is a state where $\psi^{\beta_1,\beta_2,\dots,\beta_n}$ is obtained by a contraction of a network of tensors, one tensor for each physical index, as in Fig.~\ref{f:blocking}(a).  Graph states and AKLT states are well studied examples of PEPS.  Appendix~\ref{s:peps_defs} presents a specific standard choice of tensors for describing graph states and AKLT states. 

\subsubsection{Blocking the Hamiltonian}
\label{s:blocked}

When studying Hamiltonians that have PEPS as ground states, it is often simpler to consider a Hamiltonian which has the same ground state as the physical Hamiltonian, but with interaction terms that act on a larger number of particles. Such Hamiltonians, which we will call \emph{blocked Hamiltonians}, share many properties with the physical Hamiltonian, but are often simpler to study because the spectrum of the local terms can be chosen to be very simple. Here, we define a blocked version $H^B(\delta)$ of the parent Hamiltonian $H(\delta)$.  This blocking will assist in our proof of a gap for $H(\delta)$, as we will show that proving a gap for $H(\delta)$ is equivalent to proving a gap for $H^B(\delta)$ for $\delta>0$.  The blocked Hamiltonian will possess an injectivity property that allows us to prove a gap. In addition, $H^B(\delta)$ can be rigorously proven to be gapped not only for $\delta = 0$ but for small $\delta>0$.

\begin{figure}
\centering
    \begin{subfigure}[c]{0.22\textwidth}
        \includegraphics[width=\textwidth]{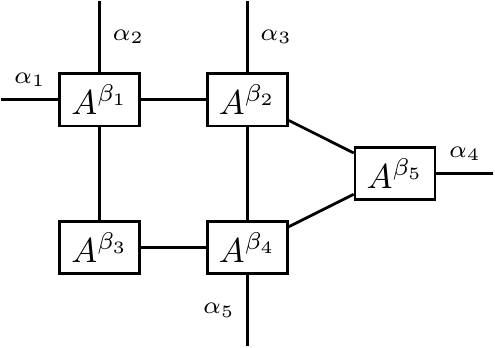}
        \caption{}
    \end{subfigure}\quad
    \begin{subfigure}[c]{0.22\textwidth}
        \includegraphics[width=\textwidth]{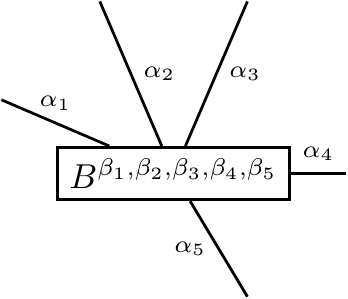}
        \caption{}
    \end{subfigure}
    \caption{A connected region of tensors can be `blocked' to form a single tensor by contracting all internal indices within the region. (a) Tensors in a five-particle region. (b) Blocked tensor.} 
    \label{f:blocking}
\end{figure}

We define the blocked Hamiltonian $H^B(\delta)$ as follows. Let $G$ be a graph, and consider the PEPS describing the state $\ket{\psi(\delta)}$ of Eq.~\eqref{e:state_definition} on this graph that interpolates between the graph state $\ket{G}$ at $\delta=0$ and the AKLT state at $\delta=1$. Given a connected region $R\subseteq V$ of $n$ vertices and $r$ outgoing edges we define the block tensor $B^{\beta_1,\beta_2,\dots,\beta_n}_{\alpha_1, \alpha_2,\dots, \alpha_r}(R,\delta)$ as the tensor obtained by contracting all virtual indices of tensors within the region and leaving outgoing virtual indices uncontracted, as illustrated in Fig.~\ref{f:blocking}(b). This blocking naturally defines a map from virtual, to physical degrees of freedom
\begin{equation}
    \hat{B}^{(R)}(\delta)=\sum_{\substack{\beta_1,\dots,\beta_n\\\alpha_1, \dots, \alpha_r}}
    B^{\beta_1,\dots,\beta_n}_{\alpha_1,\dots, \alpha_r}(R,\delta)\ketbra{\beta_1,\dots,\beta_n}{ \alpha_1,\dots, \alpha_r}\,.
    \label{e:block_map}
\end{equation}
We will say that a region $R$ of particles is injective if the associated map $\hat{B}^{(R)}(\delta)$ is injective. Generally, if the vertices of a tensor network state $\ket{\psi}$ can be partitioned into $M$ disjoint connected regions $V=\cup_{a=1}^M R_a$ such that each set $R_a$ is injective, we say that the state $\ket{\psi}$ is injective. 

In Appendix \ref{s:injective}, we show that a region $R$ for the state $\ket{\psi(\delta)}$ is injective if and only if each vertex $i\in R$ has at most one outgoing edge (i.e., an edge connected to a vertex not in $R$). We will say that $\{R_a:a=1,\dots,M\}$ is an \emph{injective covering} of $V$ if $R_a$ is injective for every $a$, $R_a\cap R_b=\emptyset$ for any $a\ne b$, and $\cup_{a=1}^M R_a=V$. Examples of lattices with injective coverings are illustrated in Fig.~\ref{f:injective_regions}, with the injective region highlighted. 
\begin{figure}
    \begin{subfigure}[c]{0.22\textwidth}
        \includegraphics[width=\textwidth]{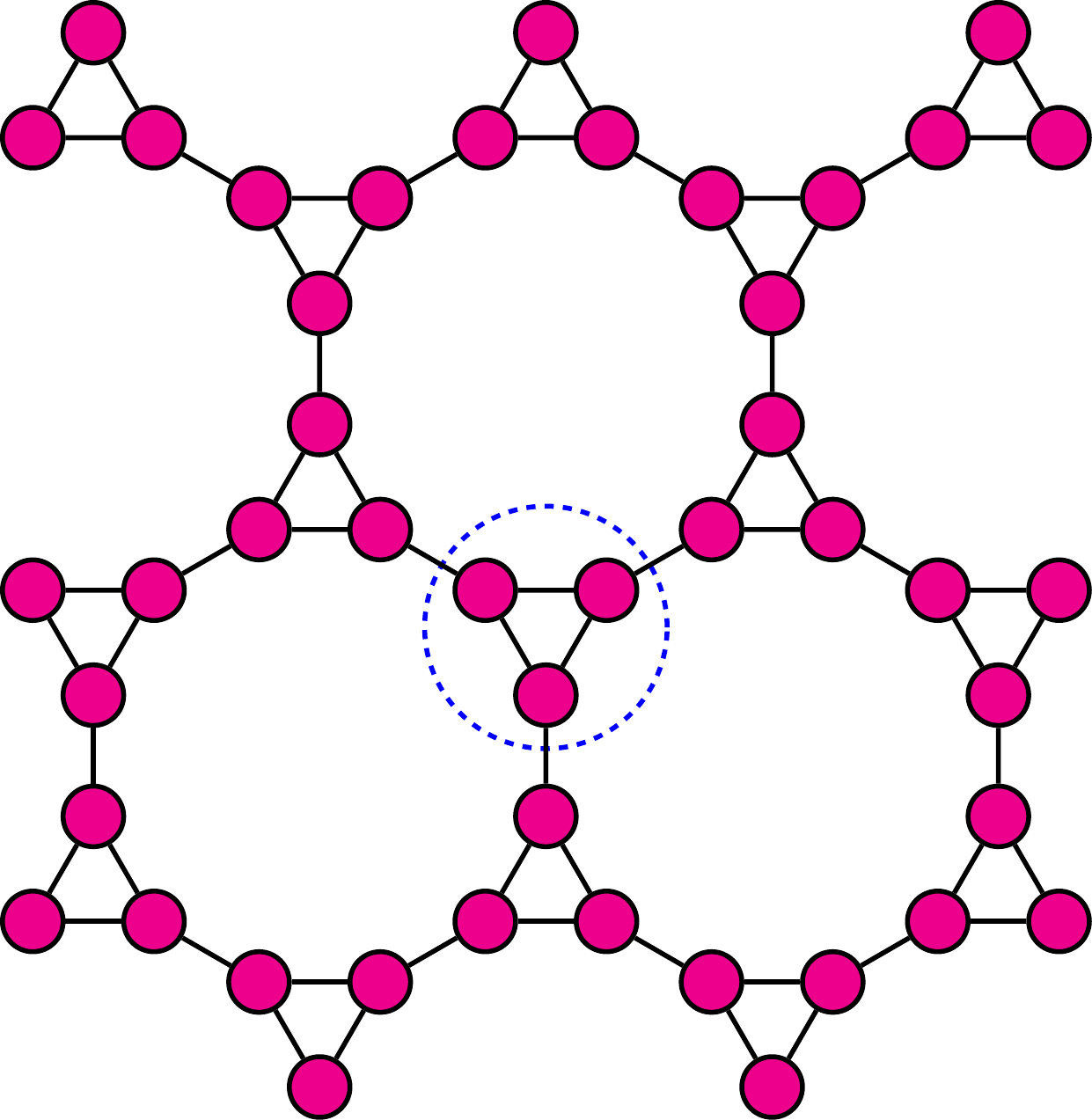}
        \caption{}
    \end{subfigure}
    \begin{subfigure}[c]{0.22\textwidth}
        \includegraphics[width=\textwidth]{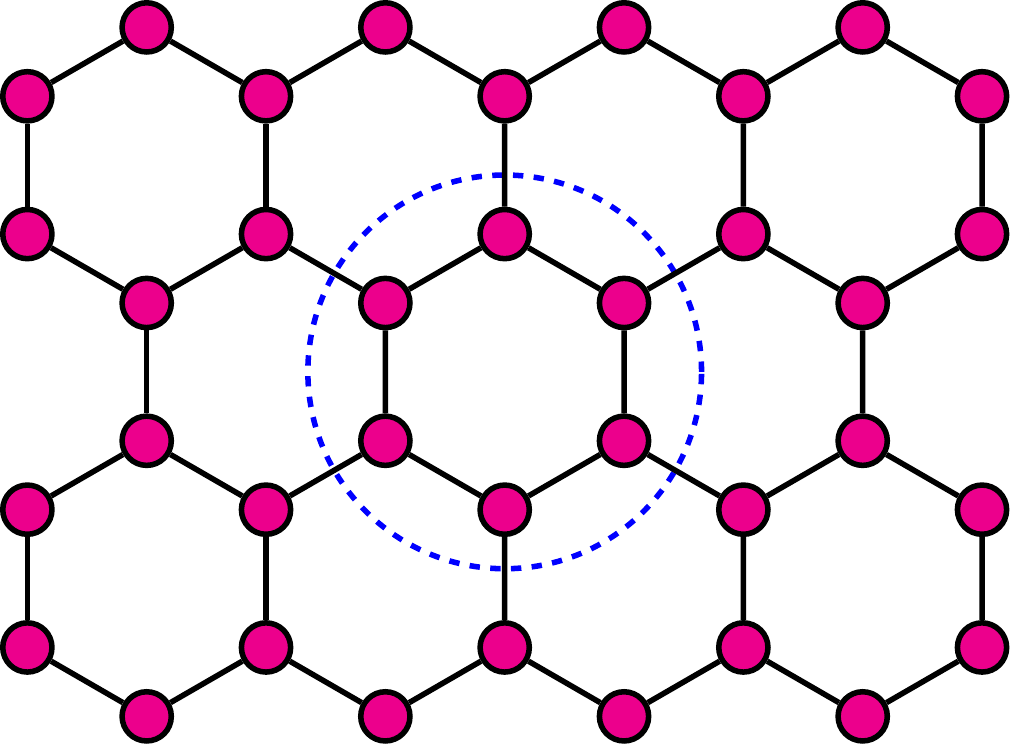}
        \caption{}
    \end{subfigure}    
    \caption{The star lattice (a) and honeycomb lattice (b) with injective regions circled. A region is injective if and only if each particle has at most one outgoing edge (i.e. connected to a particle outside the region). Furthermore, we say these lattices are coverable as both can be decomposed as as a disjoint union of such regions.}
    \label{f:injective_regions}
\end{figure}

Let $E_{a,b}\subseteq E$ be the set of edges in $G$ that connect a vertex in $R_a$ to a vertex in $R_b$. Given a partitioning of the graph into injective regions $\{R_a:a=1,\dots,M\}$, we define a coarse-grained graph $G'=(V',E')$ as follows. Each vertex in $V'$ corresponds to a region $R_a$, and two vertices in $V'$, corresponding to the regions $R_a$ and $R_b$, are connected by an edge in $E'$ if and only if $E_{a,b}$ is non-empty. 

Given the state $\ket{\psi(\delta)}$ defined on a graph with a injective covering, let $\hat{B}^{(a,b)}(\delta)$ be the injective block tensor for the region $R_a\cup R_b$, where $R_a$ and $R_b$ are neighbouring injective regions (where `neighbouring' means that they are connected in $G'$). 

We define the blocked parent Hamiltonian for $\ket{\psi(\delta)}$ as
\begin{equation}
    H^B(\delta):=\sum_{\langle a,b \rangle \in E'} \Pi^{(a,b)}(\delta)\,,
\end{equation}
where $\Pi^{(a,b)}(\delta)$ is the projector onto the orthogonal complement of the image of $\hat{B}^{(a,b)}(\delta)$, and the sum is taken over all neighbouring pairs of regions $R_a$ and $R_b$. Note that this Hamiltonian is not two-body: each term $\Pi^{(a,b)}(\delta)$ will in general act non-trivially on all particles in the region $R_a\cup R_b$. Given that the reduced density operator of $\ket{\psi(\delta)}$ is supported on the image of $\hat{B}^{(a,b)}(\delta)$, it is clear that $\ket{\psi(\delta)}$ is a ground state of $H^B(\delta)$. As is shown in Ref.~\cite{perez-garcia_peps_2008} the injectivity condition further implies that $\ket{\psi(\delta)}$ is the \emph{unique} ground state of $H^B(\delta)$. 

We thus have two Hamiltonians, $H(\delta)$ and $H^B(\delta)$, both possessing $\ket{\psi(\delta)}$ as a unique ground state. The Hamiltonian $H(\delta)$ is two-body while $H^B(\delta)$ is not. A crucial difference between these two Hamiltonians is that $H^B(\delta)$ is well defined at $\delta=0$ with a unique ground state exactly equal to the target graph state $\ket{G}$.  On the other hand, as we showed in Sec. \ref{s:finite_gap}, the two-body Hamiltonian $H(\delta)$ has an exponentially degenerate ground space at $\delta=0$. In the following section, we will use the fact that $H^B(0)$ has a unique ground state and is gapped to prove that $H(\delta)$ is gapped for sufficiently small $\delta$. 

\subsubsection{Restrictions on interaction graph}
\label{s:restrictions}

We have imposed a number of constraints on the interaction graph of the Hamiltonian, which we briefly summarise here. First, we have restricted to three-colourable graphs, for which the methods of Sec.~\ref{s:deform_definition} can be used to project the AKLT state locally to a graph state on the same graph.  To prove a gap in the thermodynamic limit, we will further require injectivity for our blocked Hamiltonian. As described in the previous section we can ensure injectivity by imposing that our graphs have an injective covering. 

As we are primarily interested in lattices (rather than arbitrary graphs) we will impose some further simplifications. We require that there exists an injective covering for the graph $\{R_a: a=1,\dots,M\}$, for which both the number outgoing edges $r$ and the number of particles $n$ in each $R_a$ is a constant independent of $a$ and the system size $N$ and that the block tensor $\hat{B}^{(a)}(\delta)$ for each $R_a$ are all identical up to a unitary acting on the output for all $\delta\ge0$. It would vastly complicate the proof if the injective covering varied with the system size or if the spectral properties of the block tensor were different for different regions.  

We will call a graph satisfying all of the above properties \emph{coverable}. All trivalent Archimedian lattices are coverable: the honeycomb $(6^3)$, the square octagon $(4,8^2)$, the cross $(4,6,12)$ and the star $(3,12^2)$. However, the square lattice (which does not have a injective covering), is not coverable. We will discuss generalisations non-coverable lattices in \ref{s:generalisations}.

\subsubsection{Proof of gap}

Here, we prove the following theorem regarding the existence of a spectral gap for $H(\delta)$, which is the central result of this section. 
\begin{theorem}
    \label{t:gap_infinite}
    Let $H(\delta)$ be the two-body Hamiltonian defined in Eq.~\eqref{e:twobody_gs_ham} on a coverable graph. Then there exists $\delta_c>0$ such that $H(\delta)$ is gapped for all $\delta$ in the open interval $(0,\delta_c)$. 
\end{theorem}
The proof of Theorem~\ref{t:gap_infinite} follows from three simple lemmas. We will first show that the blocked Hamiltonian is gapped at $\delta=0$. 
\begin{lemma}
    \label{t:block_ham_gap}
    $H^B(0)$ is gapped. 
\end{lemma}
\begin{proof}
    We prove this lemma using the fact that $\ket{\psi(0)}=\ket{G}$, and by showing that this graph state can be locally transformed to a product of Bell pairs via a product of unitaries that act only within injective regions (and not between them). Under this unitary map, the Hamiltonian can be seen to be trivially diagonalisable, commuting and thus gapped. 

Let $G$ be a coverable graph and consider its tensor description, provided in Appendix~\ref{s:gs_peps}.  Let $\{R_a: a=1,\dots,M\}$ be an injective covering of $G$.
Let $W_a:=\bigotimes_{\langle i,j\rangle\in E_a}(CZ\oplus I)_{i,j}$ be a product of $CZ$ gates acting on all pairs of neighboring particles in the region $R_a$ where each $(CZ\oplus I)_{i,j}$ acts as $CZ$  on the logical subspace $\mathcal{H}^L_i\otimes \mathcal{H}^L_j\subseteq \mathcal{H}_i\otimes \mathcal{H}$ of particles $i$ and $j$ and as the identity on the remainder of the space.
Using the identity given in Eq.~\eqref{e:graph_add_edge} on the tensor description of $\hat{B}^{(a,b)}(\delta)$ on two injective regions $R_a$ and $R_b$, we see that the image of $[W_a\otimes W_b]\hat{B}^{(a,b)}(0)$ becomes a simple tensor product. To be precise, let $R^0$ denote the interior of the region $R$, i.e., the set of vertices in $R$ that are not connected to vertices outside $R$, and let $\partial R:=R\backslash R^0$ denote the boundary of $R$. The image of $[W_a\otimes W_b]\hat{B}^{(a,b)}(0)$ is 
\begin{equation}
    \left(\bigotimes_{i\in R_a^0\cup R_b^0}\ket{+}_i\right)\left(\bigotimes_{\langle j,k \rangle\in E_{a,b}} \ket{H}_{jk}\right)\left(\bigotimes_{l\in \partial (R_a\cup R_b)} \mathcal{H}_l^L\right),
    \label{e:image}
\end{equation}
where $\ket{H}_{jk}\in \mathcal{H}^L_j\otimes \mathcal{H}^L_k$ is a two-qubit graph state on particles $j$ and $k$ and $E_{a,b}$ was as previously defined as the set of edges in $G$ that connect region $R_a$ to region $R_b$. 

Let ${\Pi^{(a,b)}}'$  be the projector onto the space orthogonal to the space defined by Eq.~(\ref{e:image}). The product form of Eq.~(\ref{e:image}) immediately implies that the set of operators $\{{\Pi^{(a,b)}}': \langle a, b \rangle \in E'\}$ pairwise commute.  It is also clear that the Hamiltonian $H':=\sum_{\langle a,b \rangle\in E'}{\Pi^{(a,b)}}'$ has a trivial unique ground state consisting of $\ket{H}$ on all connected pairs of vertices that are from separate regions, and $\ket{+}$ on the remaining vertices that are only connected to other vertices within the same region. Hence $H'$ is gapped. Finally, given that $H'$ and $H^B(0)$ are related by conjugation by $\bigotimes_{a=1}^M W_a$ we conclude that $H^B(0)$ is also gapped. 
\end{proof}

Having shown that $H^B(0)$ is gapped, we next show that $H^B(\delta)$ is gapped for a finite interval $(0,\delta_c)$.  This result can be seen to follow quite directly from stability results for canonical parent Hamiltonians~\cite{michalakis_stability_2013}, but we provide a self contained proof of this fact in Appendix~\ref{s:gap_region} based on methods used previously in Refs.~\cite{knabe_energy_1988, schuch_classifying_2011} for frustration-free Hamiltonians. The idea of the proof is to show that a PEPS parent Hamiltonian is gapped when its ground state is sufficiently close to an isometric PEPS (as is the case when $\delta$ is sufficiently small). 
\begin{lemma}
    \label{t:block_gap_region}
    There exists $\delta_c>0$ such that $H^B(\delta)$ is gapped for all $\delta$ in the interval $[0,\delta_c)$. 
\end{lemma}
\begin{proof}
    See Appendix \ref{s:gap_region}.
\end{proof}
Thus, there is an open interval on which we can prove that the blocked Hamiltonian $H^B(\delta)$ is gapped. To complete the proof of Theorem \ref{t:gap_infinite}, we need only to show that if $H^B(\delta)$ is gapped for some $\delta>0$, then the two-body Hamiltonian $H(\delta)$ is also gapped.  Here, and for the rest of the paper, we will use the convention that if $A$ and $B$ are Hermitian operators, then $A\ge B$ means that $A-B$ is a positive operator (i.e., each eigenvalue is real and greater than or equal to zero).
\begin{lemma}
    \label{t:twobody_to_block}
    For $\delta>0$, $H(\delta)$ is gapped if and only if $H^B(\delta)$ is gapped. 
\end{lemma}
\begin{proof}
    Let $\{R_a: a=1,\dots,M\}$ be an injective covering of $G$ such that the degree of each coarse-grained vertex in $V'$ (i.e., the number of outgoing edges of each $R_a$) is a constant $r$ (the existence of such a covering was an assumption of our graph). We rewrite $H(\delta)$ by grouping interaction terms according to the coarse-grained graph, as
\begin{align}
H(\delta)=\sum_{\langle a,b \rangle \in E'} h(\delta)_{ab}
\label{e:grouped_interactions}
\end{align}
where the grouped terms can be expressed in terms of this new course graining as
\begin{equation}
    h(\delta)_{ab} :=\sum_{\langle i,j \rangle\in E_a\cup E_b} r^{-1} Q(\delta)_{ij} 
    +\sum_{\langle i,j \rangle \in E_{ab}}Q(\delta)_{ij}\,,
\end{equation}
where $E_a\subseteq E$ is the set of edges contained entirely in the region $R_a$ and, as before, $E_{a,b}$ is the set of edges joining vertices in $R_a$ and $R_b$. It is straightforward to show that the sum in Eq.~\eqref{e:grouped_interactions} is equal to $H(\delta)$ as defined in Eq.~\eqref{e:twobody_gs_ham}. 

Comparing $h(\delta)_{ab}$ to the blocked interaction term $\Pi^{(a,b)}(\delta)$, one finds that the kernel of a single term $h(\delta)_{ab}$ is equal to the kernel of the projector $\Pi^{(a,b)}(\delta)$.  We will not provide a detailed proof here; it is possible to verify following the inductive proof in Ref.~\cite{niggemann_quantum_1997}, Appendix A, where it was shown that the spin-3/2 AKLT model (and deformed versions of it) has a unique ground state. Because these operators have the same kernel, we have
\begin{equation}
    \lambda_{\rm min}(\delta)\Pi^{(a,b)}(\delta)\le h(\delta)_{ab}\le\lambda_{\rm max}(\delta)\Pi^{(a,b)}(\delta)\,,
    \label{e:block_inequality}
\end{equation}
where $\lambda_{\rm min}(\delta)>0$ and $\lambda_{\rm max}(\delta)>0$ are respectively the smallest and largest non-zero eigenvalue of $h(\delta)_{ab}$.  Importantly, $\lambda_{\rm min}(\delta)>0$ and $\lambda_{\rm max}(\delta)>0$ are independent of the system size. From Eq.~\eqref{e:block_inequality}, it follows that $\lambda_{\rm min}(\delta)H^B(\delta)\le H(\delta)$.    Therefore, for any system size $N$, if $H^B(\delta)$ has gap $\Delta(\delta)$, $H(\delta)$ will have a gap of at least $\lambda_{\rm min}(\delta)\Delta(\delta)$. Hence for any fixed $\delta>0$, $H(\delta)$ is gapped if $H^B(\delta)$ is gapped. Likewise, as $\lambda_{\rm max}\iv(\delta)H(\delta)\le H(\delta)^B$, $H(\delta)$ being gapped implies  $H^B(\delta)$ is gapped.

We remark that, while $\lambda_{\rm min}(\delta)>0$ for all $\delta>0$, taking the limit gives $\lim_{\delta\rightarrow 0}\lambda_{\rm min}(\delta)=0$. Thus we cannot prove a gap at $\delta=0$ (which is unsurprising, given that we showed that the ground space becomes degenerate at this point in Sec.~\ref{s:finite_gap}).  
\end{proof}

This completes the proof of Theorem~\ref{t:gap_infinite} that, under certain assumptions on the interaction graph, $H(\delta)$ is gapped for $\delta$ sufficiently small.

\subsubsection{Generalisation to arbitrary lattices}
\label{s:generalisations}

In the proof of Theorem~\ref{t:gap_infinite}, we showed that the two body Hamiltonian $H(\delta)$ is gapped for sufficiently small $\delta>0$ provided the interaction graph is coverable (as described in Sec.~\ref{s:restrictions}). Here we will discuss obstacles to generalising this result to non-coverable graphs. 

Consider, for example, the square lattice. The state $\ket{\psi(\delta)}$ is the unique ground state of the two-body, spin-2 Hamiltonian $H(\delta)$ for $\delta>0$ on this lattice, but is it gapped? The state $\ket{\psi(\delta)}$ is not injective on this lattice, and thus Theorem \ref{t:gap_infinite} does not apply to it directly. We can nevertheless follow the proof of Theorem \ref{t:gap_infinite} and construct a parent Hamiltonian for every $\delta\ge0$ with $\ket{\psi(\delta)}$ as a unique ground state. 

First we partition the particles into $2 \times 2$ squares $\{R_a: a=1,\dots, M\}$. Let $\Pi^{(a,b,c,d)}(\delta)$ be a blocked interaction term that acts on a $2\times 2$ square of four such regions $R_a$, $R_b$, $R_c$, and $R_d$, which is defined, as usual, as the projection onto the orthogonal complement of the image of the block tensor $\hat{B}^{(R_a\cup R_b \cup R_c \cup R_d)}(\delta)$. One can show that the blocked Hamiltonian $H^B(\delta):=\sum_{\langle a, b, c, d\rangle}\Pi^{(a,b,c,d)}(\delta)$, with sum taken over all such squares, has $\ket{\psi(\delta)}$ as a unique ground state for both $\delta=0$ and $\delta>0$. It is also possible to show that Lemma~\ref{t:block_ham_gap} holds (i.e., $H^B(0)$ is gapped) and Lemma~\ref{t:twobody_to_block} holds (i.e., the two-body Hamiltonian $H(\delta)$ is gapped if and only if $H^B(\delta)$ is gapped for $\delta>0$). The problem arises when we try to interpret $H^B(\delta)$ for $\delta>0$ as a small perturbation to the gapped Hamiltonian $H^B(0)$, as is required to prove Lemma~\ref{t:block_gap_region}. One finds that there is a discontinuity of the rank of the interaction terms $\Pi^{(a,b,c,d)}(\delta)$ at $\delta=0$, specifically the rank is smaller at $\delta>0$ than at $\delta=0$. This is in contrast to the injective case, where the rank of the block tensor (and therefore also the interaction terms) is constant for all $\delta\ge 0$. As a result, the stability arguments used to prove Lemma~\ref{t:block_gap_region} and the proof methods of Appendix~\ref{s:gap_region} cannot be applied directly.  Hence, it does not seem as straight-forward to prove a gap for non-injective lattices, as for injective ones. 

\subsection{Gapped regions for trivalent lattices}
\label{s:size_gap}

Theorem~\ref{t:gap_infinite} states that for the two-body Hamiltonian $H(\delta)$ (defined on an appropriate lattice) there exists a $\delta_c>0$ such that $H(\delta)$ is gapped for all $\delta$ on the open interval $(0,\delta_c)$. Here we will explicitly compute lower bounds on $\delta_c$ on various trivalent lattices.  

Let $H(\delta)$ be defined on a graph with an injective covering $\{R_a:a=1,\dots,M\}$, where each injective region $R_a$ has the same number of outgoing edges $r$. Consider the operator $B^{(R_a)\dag}(\delta) B^{(R_a)}(\delta)$ and let $\gamma_{\rm min}(\delta)$ and $\gamma_{\rm max}(\delta)$ be its smallest and largest eigenvalues respectively. Note that the ratio $\gamma_{\rm max}(\delta)/\gamma_{\rm min}(\delta)$ is a measure of how far $B^{(R_a)\dag}(\delta) B^{(R_a)}(\delta)$ is from the identity. At $\delta=0$, this ratio achieves its minimum of 1 (i.e., $B^{(R_a)\dag}(0) B^{(R_a)}(0)$ is the identity). In Appendix~\ref{s:gap_region} we show that $H(\delta)$ is gapped if $\gamma_{\rm max}(\delta)/\gamma_{\rm min}(\delta)<\mu_0$, where 
\begin{equation}
    \label{e:gap_condition}
    \mu_0=\fr{2}+\fr{2}\sqrt{\frac{r+1}{r-1}}\,.
\end{equation}
Thus our bound depends on the number of outgoing edges in each injective region $r$. 

Note that $\mu_0>1$ is greater than 1 for any $r$, and thus there is always a region around $\delta=0$ for which this bound proves that the Hamiltonian is gapped.  

\begin{figure}[t]
    \centering
    \includegraphics[width=0.3\textwidth]{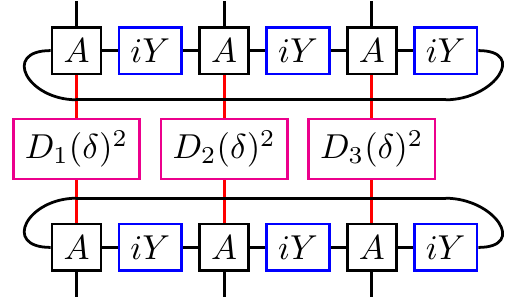}
    \caption{Expressing the operator $\hat{B}^{(R_a)}(\delta)^\dag \hat{B}^{(R_a)}(\delta)$ for the star lattice (with $r=3$) as a contraction of tensors. Where the indices at the top represent the input and the indices at the bottom represent the output. Physical indices in this figure are represented by red edges. The Hamiltonian $H(\delta)$ is gapped when the eigenvalues of this operator satisfy Eq. \eqref{e:gap_condition}. }
    \label{f:doubled_tensor}
\end{figure}
We can evaluate the operator $B^{(R_a)\dag}(\delta) B^{(R_a)}(\delta)$ by contracting a tensor network, illustrated in Fig. \ref{f:doubled_tensor}. For $r=3,4,5,6$ we have $\mu_0^{-1} \approx 0.828,\ 0.873,\ 0.899,\ 0.916$ respectively corresponding respectively to $\delta_c\approx 0.28,\ 0.13,\ 0.12,\ 0.08$ and the Hamiltonian $H(\delta)$ will be gapped for all $\delta \in (0, \delta_c)$. Our upper bound on the size of the gapped region decreases as a function of $r$, the number of outgoing edges in an injective region.

\section{Conclusion}
\label{s:conclusion}

We have shown that methods from quantum information theory, in particular from the study of graph states and AKLT states for quantum computation, can be useful for proving spectral properties of 2D antiferromagnetic spin models.  In particular, for a family of 2D antiferromagnetic spin models obtained by a one-parameter deformation of the 2D AKLT model defined here, we have proven that there exists a spectral gap separating the ground state from the first excited state in the thermodynamic limit for a range of parameter space of this family.  

An implication of our results is that this family of models has many of the desired properties for quantum computation, in particular that there exists a finite region in the parameter space that is both computationally universal and gapped.  The existence of a gap may allow for efficient methods to prepare such resource states for quantum computation~\cite{verstraete_quantum_2009, kraus_preparation_2008-1, johnson_general_2015}.

It would be desirable to prove that the gapped region extends to $\delta=1$ for the class of lattices defined here, where $H(\delta{=}1)$ is the AKLT Hamiltonian $H^A$.  This would settle the long-standing conjecture regarding the existence of a spectral gap for the spin-$3/2$ 2D AKLT model. While the region we prove to be gapped unfortunately does not extend to $\delta=1$, there are various ways by which proofs of a larger gapped region may be sought.  For instance in, mixing properties of Markov chains were used in Ref.~\cite{perez-garcia_peps_2008} to show that the Ising PEPS parent Hamiltonian is gapped up to the critical point.  Other methods from statistical physics, e.g., using `martingales' \cite{spitzer_improved_2003} may also be used to improve our result. 

\begin{acknowledgements}
  We thank Gavin Brennen for discussions.  We acknowledge support from the ARC via project number DP130103715 and the Centre of Excellence in Engineered Quantum Systems (EQuS), project number CE110001013. ASD is partially funded by Canada's NSERC and the Canadian Institute for Advanced Research.
\end{acknowledgements}

\appendix 
\section{PEPS definitions}
\label{s:peps_defs}
In this section we will provide simple tensor-network descriptions of AKLT states and graph states and for the interpolating path between the two.

\subsection{AKLT states as PEPS}
\label{s:aklt_tensor_def}

An AKLT state may be defined on a graph $G=(V,E)$ as follows. At each vertex $i$ of the graph place a spin-$S_i$ particle, where $S_i=(d_i-1)/2$ and $d_i$ is the degree of vertex $i\in V$. A basis for the Hilbert space of each particle is given by $\{\ket{M}:M=S_i,S_i-1,\dots,-S_i+1, -S_i\}$, corresponding to the eigenstates of the $z$-component of the spin operator $S^z$ . For each vertex $i$ let $A^{\beta_i}_{\alpha_1,\alpha_2,\dots \alpha_{d_i}}$ be a tensor of Clebsh-Gordan coefficients for a set of $d_i$ spin-$1/2$ particles coupling to total spin-$S_i$, where $\beta_i=M$ is the total spin along the $z$-axis and $\alpha_k$ is the component of the $k$-th constituent spin-1/2 along the $z$-axis. For example, if $S_i=3/2$, $A^\beta_{\alpha_1,\alpha_2,\alpha_3}$ will have the following non-zero entries
\begin{align*}
    A^{\frac{3}{2}}_{000}=A^{-\frac{3}{2}}_{111}&=1\,,\\
    A^{\frac{1}{2}}_{001}=A^{\frac{1}{2}}_{010}=A^{\frac{1}{2}}_{100}&=1/\sqrt{3}\,,\\
    A^{-\frac{1}{2}}_{011}=A^{-\frac{1}{2}}_{110}=A^{-\frac{1}{2}}_{101}&=1/\sqrt{3}\,,
\end{align*}
where $0$ and $1$ subscripts refer to the spin eigenstates $m_z=1/2$ and $m_z=-1/2$ respectively.

Given the $A$ tensors, the resulting tensor network for the AKLT state is defined according to Fig. \ref{f:aklt_graph_tensors}(a). An $iY$ matrix is placed on each edge connecting neighbouring $A$ tensors. This is due to the fact that the bonds used to define the AKLT PEPS are singlets $1/\sqrt{2}(1\otimes iY) (\ket{00}+\ket{11})$, rather than $1/\sqrt{2}(\ket{00}+\ket{11})$,  and gives the AKLT state its antiferromagnetic character. (Note that we do not need to specify which index of these edge $iY$ tensors is contracted with which neighbouring tensor, as swapping these indices only results in an overall ignorable phase of $-1$.) Note that we may contract tensors without a physical index (the $iY$ tensors in this case) with their neighbours such that each tensor has one physical index, thus satisfying the usual definition of a PEPS.

\subsection{Graph states as PEPS}
\label{s:gs_peps}

Graph states may also be represented as simple tensor networks. Let $G=(V,E)$ be a graph defining a graph state, according to Eq.~\eqref{e:gs_definition} and let $d_i$ be the degree of vertex $i$. At each vertex $i\in V$ place a spin-1/2 particle and define a tensor $C$ with $d_i$ virtual indices and for which there are only two non-zero entries 
\begin{equation}
    C^{0}_{00\dots 0}=C^{1}_{11\dots 1}=1\,,
    \label{e:c_def}
\end{equation}
where the labels $0$ and $1$ for physical and virtual indices may be regarded as labels for spin $1/2$ and $-1/2$ states respectively. Applying a $CZ$ gate to the physical index of any two disconnected $C$ tensors is equivalent to adding an additional virtual index to each $C$ and contracting a Hadamard matrix $H=\rt{2}(X+Z)$ between them. 

Graphically this can be represented as 
\begin{equation}
    \vcenter{\mbox{\includegraphics[width=0.42\textwidth]{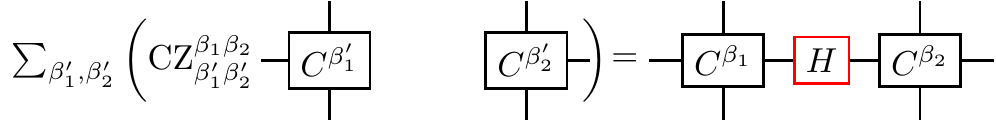}}}
    \label{e:graph_add_edge}
\end{equation}
where ${\rm CZ^{\beta_1,\beta_2}_{\beta_1',\beta_2'}}$ is the tensor of the ${\rm CZ}$ in the standard computational basis with $\beta_1',\beta_2'$ being input indices and $\beta_1, \beta_2$ being output. Although we have illustrated this with with $C$ tensors that initially have three virtual indices, this identity will hold for any number of initial virtual indices (including zero). Note that as $H^T=H$, there is no need to specify which index of $H$ is contracted with which neighbour.  The tensor network for a general graph state therefore consists of a $C$ tensor at each vertex and a $H$ contracted between neighbouring $C$ tensors on each edge, as illustrated in Fig.~\ref{f:aklt_graph_tensors}(b).

\subsection{The family of states $\ket{\psi(\delta)}$ as PEPS}

Finally, we will define the tensor network for the family of state $\ket{\psi(\delta)}$ defined in Sec.~\ref{s:deform_definition} interpolating between between the AKLT state (at $\delta=1$) and a graph state (at $\delta=0$). As these states can be obtained by applying a product of single-particle operators $\bigotimes_i D_i(\delta)$ to the AKLT state, the tensor network of $\ket{\psi(\delta)}$ can be obtained simply by replacing the $A$ tensors of the AKLT state with the $\delta$-dependent tensors $A(\delta)^{\beta_i}_{\alpha_1,\alpha_2,\dots \alpha_{d_i}}:=\sum_{\beta_i'}D_i(\delta)^{\beta_i \beta_i'}A^{\beta_i'}_{\alpha_1,\alpha_2,\dots \alpha_{d_i}}$. We show in Appendix \ref{s:local_conversion} that the tensor network for the state $\ket{\psi(0)}$ is equivalent, up to local unitaries, to the tensor network of a graph state.

We remark that the tensors we have used to define AKLT states, graph states and $\ket{\psi(\delta)}$ have only real entries. 

\begin{figure}
    \centering
    \begin{subfigure}[c]{0.23\textwidth}
        \includegraphics[width=\textwidth]{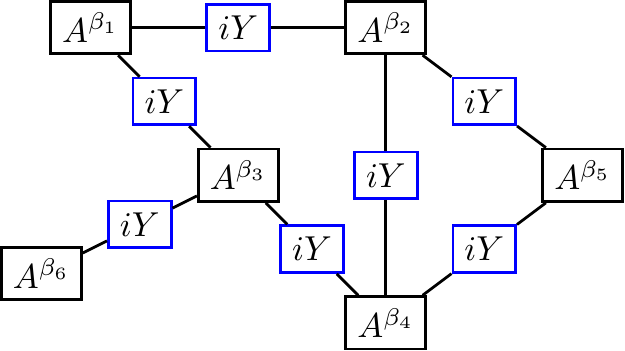}
        \caption{}
        \label{f:aklttensor}
    \end{subfigure}
    \begin{subfigure}[c]{0.23\textwidth}
        \includegraphics[width=\textwidth]{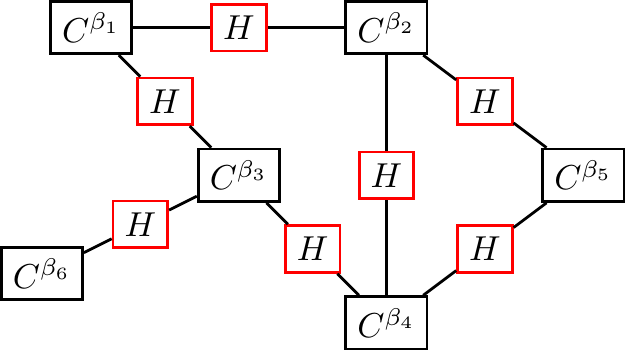}
        \label{f:graph_tensors}
        \caption{}
    \end{subfigure}    
    \caption{PEPS for an AKLT state (a) and a graph state (b) defined on a particular graph of 6 vertices, where the $A$ and $C$ tensors are defined in the text. Note that the coloured tensors do not have a physical index, and can be absorbed into their neighbours if desired. }
    \label{f:aklt_graph_tensors}
\end{figure}

\section{Local conversion of AKLT states to graph states}
\label{s:local_conversion}
Here we will provide a proof of the fact that, on three-colourable graphs, AKLT states can be converted to graph states  by applying a rank-2 projector to each particle. The result was originally proved in Ref.~\cite{wei_affleck-kennedy-lieb-tasaki_2011} for the particular case of spin-$3/2$ particles using the stabilizer formalism. Here we will provide an alternative proof using tensor networks. 

Let $G$ be a three-colourable graph and let $\ket{\rm AKLT}$ be the AKLT state defined on this graph, according to Appendix \ref{s:aklt_tensor_def}. Let $\{c_i\in \{x,y,z\}:i=1\dots N\}$ be a three-colouring of the vertices of $G$, in the sense that $c_i\ne c_j$ if $i$ and $j$ are neighbours in $G$. 

Define $P^x$, $P^y$ and $P^z$ as in Eq. \eqref{e:projector_def}. Then the AKLT state can be transformed into a graph state by these local operators in the following sense. 
\begin{lemma} The following relation holds:
    \begin{equation}
        \bigotimes_{i=1}^N P^{c_i} \ket{{\rm AKLT}}\propto\bigotimes_{i=1}^NZ(\theta_i)\ket{G}\,,
        \label{e:aklttograph}
    \end{equation}
    where $\ket{G}$ is the graph state as defined in Appendix~\ref{s:gs_peps}, $Z(\theta)={\rm diag}(1, \exp(i \theta))$ is a local $z$-rotation by $\theta$ and where $\theta_i \in \{0,\pi/2, \pi, 3\pi/4\}\ \forall\ i$. 
\end{lemma}
\begin{proof}Applying $\bigotimes_{i=1}^N P^{c_i}$ to the AKLT state attaches one of $P^x, P^y, P^z$ to the physical index of each tensor. It is straightforward to show that applying $P^z$ transforms the AKLT tensor $A$ to the graph state tensor $C$, i.e. $\sum_{\beta'}(P^z)^{\beta,\beta'}A^{\beta'}_{\alpha_1\alpha_2\dots \alpha_n}=C^{\beta}_{\alpha_1\alpha_2\dots \alpha_n}$. Now let $e^{i\theta \vec{r}\cdot \vec{S}}$ be the spatial rotation by angle $\theta$ about the axis $\vec{r}$, where $\vec{S}=(S^x,S^y,S^z)$ is the vector of spin operators.  (Note that we have adopted a different convention in our above definition of the $z$-rotation of a spin-1/2 particle $Z(\theta)$, so that, for example $Z(\pi)$ equals Pauli $Z$, however these two definitions differ only up to an overall phase.)  From the symmetry of the $A$ tensors, applying the same spatial rotation to the physical index and all virtual indices leaves $A$ invariant. Furthermore, applying either $P^x$ or $P^y$ is equivalent to applying a rotation to the physical index, followed by applying $P^z$.  From this, one can see that, up to an overall phase, $\sum_{\beta'}(P^x)^{\beta,\beta'}A^{\beta'}_{\alpha_1\alpha_2\dots \alpha_n}$ and  $\sum_{\beta'}(P^y)^{\beta,\beta'}A^{\beta'}_{\alpha_1\alpha_2\dots \alpha_n}$ are respectively 
\begin{align}
    \sum_{\alpha_1',\alpha_2',\dots,\alpha_n'}H_{\alpha_1,\alpha_1'}H_{\alpha_2,\alpha_2'}\cdots H_{\alpha_n,\alpha_n'}C^{\beta}_{\alpha_1'\alpha_2'\dots \alpha_n'}\notag\,,\\
    \sum_{\alpha_1',\alpha_2',\dots,\alpha_n'}H_{\alpha_1,\alpha_1'}'H_{\alpha_2,\alpha_2'}'\cdots H_{\alpha_n,\alpha_n'}'C^{\beta}_{\alpha_1'\alpha_2'\dots \alpha_n'}\label{e:hadamard_insertion}\,.
\end{align}
where $H=\rt{2}(X+Z)$ and $H'=\rt{2}(Y+Z)$ are complex Hadamard matrices. Applying the projectors $\bigotimes_{i=1}^N P^{c_i}$ to $\ket{\rm AKLT}$ thus converts all $A$ tensors to $C$ tensors and attaches $H$ and $H'$ to their virtual indices according to which projector ($P^x$  $P^y$ or $P^z$) was applied. This is illustrated for a three particle AKLT state in Fig \ref{f:aklt_to_gs}(a)-(b). 
\begin{figure}
    \centering
    \begin{subfigure}[c]{0.22\textwidth}
        \includegraphics[width=\textwidth]{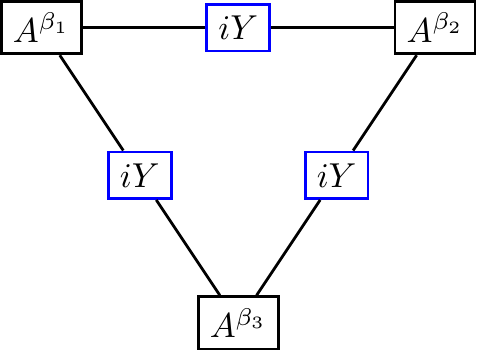}
        \caption{}
    \end{subfigure}\quad\quad
    \begin{subfigure}[c]{0.22\textwidth}
        \includegraphics[width=\textwidth]{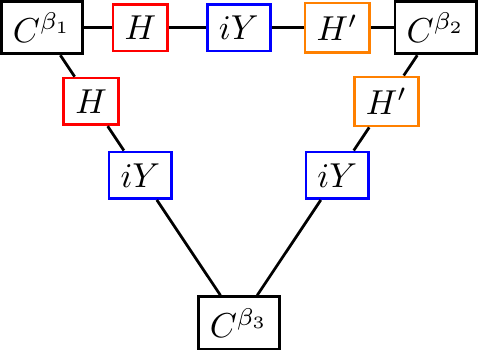}
        \caption{}
    \end{subfigure}\quad\quad
    \begin{subfigure}[c]{0.22\textwidth}
        \includegraphics[width=\textwidth]{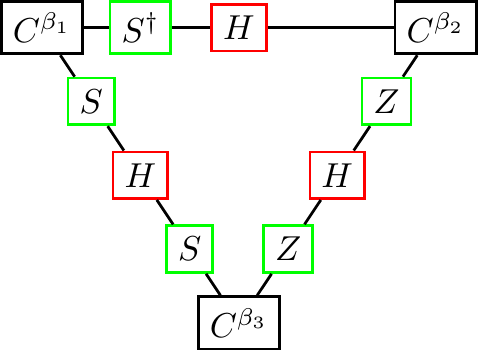}
        \caption{}
    \end{subfigure}\quad\quad
    \begin{subfigure}[c]{0.22\textwidth}
        \includegraphics[width=\textwidth]{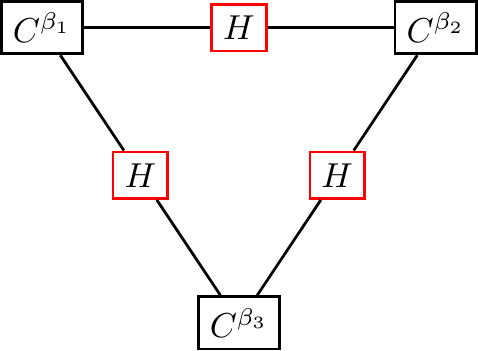}
        \caption{}
    \end{subfigure}
    \caption{Illustrating how to convert AKLT states to graph states with single-particle projectors. We consider an AKLT state on a triangle where a different projector is applied to each particle. (a) AKLT state on a triangle. (b) State after applying $P^x$, $P^y$, $P^z$ to particles 1, 2 and 3 respectively. (c) Rewriting each edge matrix as $H$ multiplied on the left and right by $z$-rotations. (d) Absorbing all $z$-rotations into the $C$'s. They will appear only on the physical indices (not pictured).  }
    \label{f:aklt_to_gs}
\end{figure}
To prove the lemma, we show that this tensor network is equivalent to the tensor network of a graph state described in Appendix~\ref{s:gs_peps} up to local $z$-rotations. We will use the fact that for the $C$ tensor, $Z(\theta)$ applied to any virtual index has the same effect as $Z(\theta)$ on the physical index, i.e.,
\begin{equation}
    Z(\theta)_{\alpha_k \alpha_k'}C^{\beta}_{\alpha_1'\alpha_2\dots \alpha_{k-1}\alpha_k'\alpha_{k+1} \dots \alpha_n}=Z(\theta)_{\beta, \beta'}C^{\beta'}_{\alpha_1\alpha_2\dots \alpha_n}
    \label{e:absorb_z}
\end{equation}
for any $k$. 

As seen in Fig \ref{f:aklt_to_gs}(b), after applying the projectors \eqref{e:projector_def} to the state and using Eq.~\eqref{e:hadamard_insertion}, we are left with a product of matrices along edges of the tensor network between the $C$ tensors. By assumption, if $i$ and $j$ are neighbouring vertices, $c_i\ne c_j$. Thus there are only three distinct edges we must consider $xy$, $xz$ and $yz$. These edges have the matrices $iHYH'$, $iHY$, and $iYH'$ applied on them respectively. These can all be expressed as $H$ multiplied on the left and right by $z$-rotations and by an overall phase, specifically 
\begin{align}
    iHYH'&=e^{\frac{3\pi}{4}i}S^\dag H\,, \notag\\
    iHY&=ZHZ\,,\notag\\
    iYH'&=iSHS\,,
\end{align}
where $S=Z(\pi/2)$. Using \eqref{e:absorb_z} all $z$-rotations can be moved to the physical indices of the $C$. From this we see that the tensor network is equivalent to the tensor network of a graph state described in Appendix~\ref{s:gs_peps}, up to $z$-rotations on the physical indices.  
\end{proof}

We remark that from the definition of $D_i(\delta)$ in Eq.~\eqref{e:def_deformations}, we have $P^{c_i}D_i(\delta)\propto P^{c_i}$, and therefore the result also holds for all deformed states $\ket{\psi(\delta)}$. 
\begin{corollary}
    \label{t:defaklt_to_gs}
    Let $G$ be a three-colourable graph, and $\{c_i\in \{x,y,z\}:i=1\dots N\}$ be a three-colouring of $G$. Define the corresponding one-parameter family of states $\ket{\psi(\delta)}$ according to \ref{s:deform_definition}. Then 
    \begin{equation}
        \bigotimes_{i=1}^N P^{c_i} \ket{\psi(\delta)}\propto\bigotimes_{i=1}^NZ(\theta_i)\ket{G}\,,
    \end{equation}
    for all $\delta\ge 0 $.
\end{corollary}

\section{Injective regions}
\label{s:injective}
Here we will prove a necessary and sufficient condition for a region of particles in the tensor network of $\ket{\psi(\delta)}$ to be injective. 
\begin{lemma}
    A region $R$ in the tensor network state $\ket{\psi(\delta)}$ for any $\delta\ge0$ is injective if and only if each vertex in $R$ is connected to at most one vertex outside $R$. 
\end{lemma}
\begin{proof}
    Let $R$ be a region for which every vertex is connected to at most one vertex outside $R$, and $\hat{B}^{(R)}(\delta)$ be block tensor for this region (as defined in Sec.~\ref{s:blocked}). First we will show show that this map is injective, i.e., $R$ is injective. We showed in Appendix~\ref{s:local_conversion} that we can transform the tensors of the state $\ket{\psi(\delta)}$ to those of $\ket{\psi(0)}=\ket{G}$ via a set of rank-2 projectors. Specifically, we have that $\left(\otimes_{i\in R} P_i\right) \hat{B}^{(R)}(\delta) \propto \hat{B}^{(R)}(0)$. Therefore the rank of $\hat{B}^{(R)}(\delta)$ is equal to or greater than the rank of $\hat{B}^{(R)}(0)$ for all $\delta\ge0$. Thus showing $\hat{B}^{(R)}(0)$ is injective will imply injectivity of $\hat{B}^{(R)}(\delta)$ for all $\delta\ge0$.

To show that $\hat{B}^{(R)}(0)$ is injective, we will apply a unitary to its output (which will not change its injectivity) and transform it into a form that is clearly injective. We define a unitary $W:=\bigotimes_{\langle i,j\rangle\in E}({\rm CZ}\oplus I)_{i,j}$, as a product of ${\rm CZ}$ gates acting on all pairs of neighboring particles in the region $R$ where each $({\rm CZ}\oplus I)_{i,j}$ acts as ${\rm CZ}$ on the logical subspace $\mathcal{H}^L_i\otimes \mathcal{H}^L_j\subseteq \mathcal{H}_i\otimes \mathcal{H}$ of particles $i$ and $j$ and as the identity on the remainder of the space. (Given, however, that the image of $\hat{B}^{(R)}(0)$ is contained entirely in the logical space $\bigotimes_{i\in R}\mathcal{H}^L_i$, the action of $W$ on the remainder of the space is not actually important.) 

We have illustrated the tensor network describing the map $B^{(R)}(0)$ before and after applying $W$ in Fig.~\ref{f:graph_disentangle}. Using Eq.~\eqref{e:graph_add_edge} to remove edges,  we see that $WB^{(R)}(0)$ describes the identity map from each outgoing virtual degree of freedom into the logical subspace $\mathcal{H}^L_i$ of the particle $i$ it is incident to, with all interior particles (those without an outgoing edge) placed in the $\ket{+}$ state.  The identity map is clearly injective, therefore $B^{(R)}(\delta)$ is injective for any $\delta\ge0$.

\begin{figure}
    \centering
    \begin{subfigure}[c]{0.23\textwidth}
        \includegraphics[width=\textwidth]{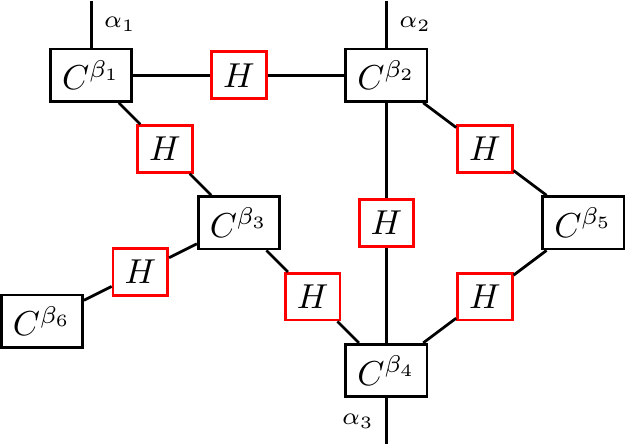}
        \caption{}
    \end{subfigure}
    \begin{subfigure}[c]{0.23\textwidth}
        \includegraphics[width=\textwidth]{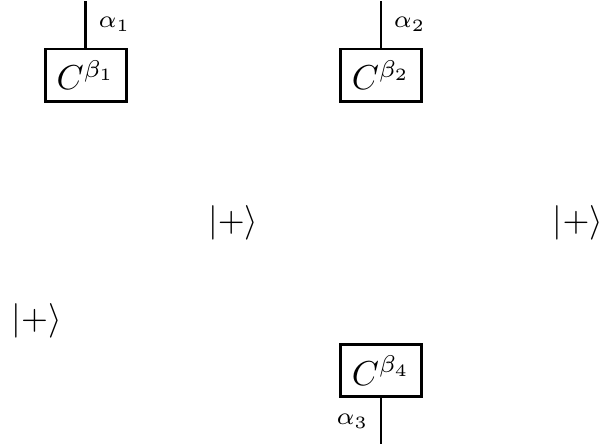}
        \caption{}
    \end{subfigure}    
    \caption{Disentangling a region in the tensor network of a graph state by applying the unitary $W$. One can clearly see that the corresponding map from outgoing virtual indices to physical indices is injective. (a) A particular region in the tensor network of a graph state with three outgoing indices. (b) The same region after the unitary $W$ has been applied to physical indices.}
    \label{f:graph_disentangle}
\end{figure}

It remains to show that having at most one outgoing index per particle is also a necessary condition for injectivity. Consider a region $R$ containing a vertex with two outgoing indices $\alpha_1$ and $\alpha_2$, say. Let $A^{\beta_i}_{\alpha_1,\alpha_2,\dots \alpha_{d_i}}$ be the tensor corresponding to this vertex. From its definition, $A$ is symmetric under any permutation of the virtual indices. Thus the input $\ket{01}-\ket{10}$ on indices $\alpha_1$ and $\alpha_2$ of the map $B^{(R)}(1)$ will map to zero. Thus $B^{(R)}(1)$ is not injective. As $B^{(R)}(\delta)\propto\bigotimes_{i\in R} D_i(\delta)B^{(R)}(1)$ for all $\delta\ge 0$, $B^{(R)}(1)$ being non-injective implies $B^{(R)}(\delta)$ is non-injective for all $\delta\ge0$.
\end{proof}

\section{Stability of gapped region}
\label{s:gap_region}
Here we will provide a proof of Lemma \ref{t:block_gap_region}, which states that under the assumptions of Theorem \ref{t:gap_infinite}, there exists $\delta_c>0$ such that $H^B(\delta)$ is gapped for all $\delta\in [0, \delta_c)$.

We make use of an argument, previously used in Refs.~\cite{knabe_energy_1988, schuch_classifying_2011} for frustration-free Hamiltonians. Let $H=\sum_{\langle j, k\rangle }h_{jk}$ be a two-body Hamiltonian acting on $N$ particles, and for simplicity we will assume that the interaction graph has constant degree $r$. Assume that the Hamiltonian is frustration-free and has a unique ground state with zero energy. If there exists a constant $\Delta>0$ such that $H^2\ge \Delta H$ for all system sizes, then $H$ is gapped.

Without loss of generality, we assume that the Hamiltonian is scaled such that $h_{kl}^2\ge h_{kl}$ for all interaction terms. 
A sufficient condition for the gap is that there exists a $\Delta>0$ such that for every particle $c$ and for any pair of distinct interaction terms $h_{ic}$, $h_{cj}$ that act non-trivially on $c$ we have 
\begin{equation}
    h_{ic}h_{cj}+h_{cj}h_{ic}\ge -\frac{1-\Delta}{2(r-1)}(h_{ic}+h_{cj})\,.
    \label{e:sufficient_condition}
\end{equation}

To see this, we expand $H^2$ to get
\begin{align}
    H^2=\left(\sum_{\langle j, k\rangle } h_{jk}^2\right)+\left(\sum_{\langle k, l\rangle,\langle p, q\rangle:\langle k, l\rangle\ne\langle p, q\rangle  } h_{kl}h_{pq}\right)\,.
\end{align}
The first term is greater than or equal to $H$. For the second term, we can split the sum into terms where $h_{kl}$ and $h_{pq}$ overlap at one site, and terms where the interactions do not overlap (disjoint). The disjoint term is clearly positive, thus 
\begin{equation}
    H^2\ge H+\left(\sum_{\langle k, l\rangle\ne\langle p, q\rangle\, {\rm overlapping}} h_{kl}h_{pq}\right)\,.
    \label{e:h2inequality}
\end{equation}
Using condition Eq.~\eqref{e:sufficient_condition} in Eq.~\eqref{e:h2inequality}, we find that the right-hand side reduces to $\Delta H$. Therefore $H$ is gapped.

Next, we show that if a Hamiltonian satisfying the gap property Eq.~\eqref{e:sufficient_condition} is deformed by a sufficiently small amount, it will remain gapped. 

\begin{lemma}
    \label{t:stab_gap}
    Let $H=\sum_{\langle i,j \rangle} h_{ij}$ be a frustration-free Hamiltonian with a unique ground state on a graph of constant degree $r$, such that $\{h_{ij}\}$ are projectors satisfying the gap condition \eqref{e:sufficient_condition} for some $0<\Delta\le1$. Define the Hamiltonian $H'=\sum_{\langle i,j \rangle} h'_{ij}$ with interaction terms
\begin{equation}
    h'_{ij}=\left(\Lambda_i\otimes\Lambda_j\right)h_{ij}\left(\Lambda_i\otimes\Lambda_j\right)\,,
    \label{e:def_interaction_terms}
\end{equation}
for a set of invertible positive operators $\{\Lambda_i\}$ which satisfy $\Lambda_i\ge I$ and let $\mu\ge1$ be the largest eigenvalue of $\Lambda_i^2$ (which we assume is independent of $i$). Define 
\begin{equation}
   \Delta':=1-\mu\left[1-\Delta+2(r-1)(\mu-1)\right]\,.
   \label{e:deformed_delta}
\end{equation}
Then $H'$ is gapped when $0<\Delta' $.
\end{lemma}
\begin{proof}
Note that $\Delta'\le \Delta$ for all $\mu\ge 1$. We claim that, 
\begin{equation}
    h_{ij}'h_{jk}'+h_{jk}'h_{ij}'+ \frac{1-\Delta'}{2(r-1)}(h_{ij}'+h_{jk}')\ge0\,.
    \label{e:deformed_gap}
\end{equation}
Therefore if there exists $\mu$ such that $0<\Delta'$, the deformed Hamiltonian will itself satisfy Eq.~\eqref{e:sufficient_condition} and be gapped. We define $\Theta=\Lambda^2-1\ge0$, which satisfies $(\mu-1)I\ge \Theta$. To prove Eq.~\eqref{e:deformed_gap}, multiply the left hand side on the left and right by  $\Lambda_i\iv\otimes\Lambda\iv_j\otimes\Lambda\iv_k$ and and simplify as follows
\begin{align}
    &h_{ij}\Lambda_j^{2}h_{jk}+h_{jk}\Lambda_j^{2} h_{ij}+\frac{1-\Delta'}{2(r-1)}(h_{ij}\otimes \Lambda_k^{-2}+\Lambda_i^{-2}\otimes h_{jk})\,,\notag\\
    &\ge h_{ij}\Lambda_j^{2}h_{jk}+h_{jk}\Lambda_j^{2} h_{ij}+\frac{1-\Delta'}{2\mu(r-1)}(h_{ij}\otimes I + I\otimes h_{jk})\,.\notag\\
    &=h_{ij}\Lambda_j^{2}h_{jk}+h_{jk}\Lambda_j^{2} h_{ij}\notag\\
    &\quad+\fr{2(r-1)}\left[1-\Delta+2(r-1)(\mu-1)\right](h_{ij}\otimes I + I\otimes h_{jk})\,,\notag\\
    &=h_{ij}\Theta_j h_{jk}+h_{jk}\Theta_j h_{ij}\notag\\
    &\quad+\left[h_{ij} h_{jk}+h_{jk}h_{ij}+\frac{1-\Delta}{2(r-1)}(h_{ij} + h_{jk})\right]\notag\\
    &\quad+(\mu-1)(h_{ij} + h_{jk})\,,\notag\\
    &\ge h_{ij}\Theta_j h_{jk}+h_{jk}\Theta_j h_{ij}+(\mu-1)(h_{ij} +  h_{jk})\,,\notag\\
    &= h_{ij}\Theta_j h_{jk}+h_{jk}\Theta_j h_{ij}+h_{ij}(\mu-1)h_{ij}+h_{jk}(\mu-1)h_{jk} \,,\notag\\
    &\ge h_{ij}\Theta_j h_{jk}+h_{jk}\Theta_j h_{ij}+h_{ij}\Theta_j h_{ij}+h_{jk}\Theta_j h_{jk}\,,\notag\\
    &=(h_{ij}+h_{jk})\Theta_j (h_{ij}+h_{jk})\ge 0\,.\notag
\end{align}
By multiplying left and right by $\Lambda_i\otimes\Lambda_j\otimes\Lambda_k$, we obtain Eq.~\eqref{e:deformed_gap}. 
\end{proof}
If the terms $h_{ij}$ are commuting, then Eq.~\eqref{e:sufficient_condition} will clearly hold for any $0<\Delta\le 1$ as the left-hand side will be positive, while the right-hand side will be negative. In this case, we have the following:
 
\begin{corollary}
    \label{t:commuting_gap}
    Let $H$ and $H'$ be Hamiltonians defined as in Lemma~\ref{t:stab_gap} with the additional property that interaction terms $h_{ij}$ commute. Then $H'$ is gapped if  
    \begin{equation}
        \label{e:gap_inequality}
        \mu < \fr{2} + \fr{2} \sqrt{\frac{r+1}{r-1}}\,.
    \end{equation}
\end{corollary}
\begin{proof}
    As $H$ is commuting, we can set $\Delta=1$ in Eq.~\eqref{e:deformed_delta}. Equation~\eqref{e:gap_inequality} then follows by solving $0<\Delta'$ for $\mu$ and restricting to $\mu\ge1$.
\end{proof}

We will now apply this result to our specific Hamiltonian. Consider the map $\hat{B}^{(R_a)}(\delta)$ from virtual to physical degrees of freedom on an injective region $R_a$. We can polar decompose $\hat{B}^{(R_a)}(\delta)=Q_a(\delta) W_a(\delta)$ where $W_a(\delta)$ is an isometry onto the image of $\hat{B}^{(R_a)}(\delta)$ and $Q_a(\delta)$ is an invertible positive operator that maps the image of $W_a(\delta)$ onto itself. As the action of $Q_a(\delta)$ on $({\rm img }\,W_a(\delta))^\perp$ is arbitrary, we assume that $Q_a(\delta)$ is a square matrix that acts only the ${\rm img}\,W_a(\delta)$. 

Let $\Lambda_a(\delta)=\gamma(\delta)^{-1} Q_a(\delta)\iv$, where $\gamma(\delta)$ is the smallest eigenvalue of $Q(\delta)\iv$ (which we have included to ensure that $\Lambda_a(\delta)\ge1$ as is required in Lemma \ref{t:stab_gap}).  Consider the block tensor $\hat{B}^{(a,b)}(\delta)$ for two neighbouring injective regions $R_a$ and $R_b$ and let $h_{a,b}(\delta)$ be the projection onto the orthogonal complement of the image of $(\Lambda_a(\delta)\otimes\Lambda_b(\delta)) \hat{B}^{(a,b)}(\delta)$. The map $\Lambda_a(\delta) \hat{B}^{(a)}(\delta)\propto W_a(\delta)$ is proportional to an isometry and therefore is proportional to $\hat{B}^{(a)}(0)$ up to left-multiplication by a unitary. Therefore, the map $(\Lambda_a(\delta)\otimes\Lambda_b(\delta)) \hat{B}^{(a,b)}(\delta)$ is equivalent to $\hat{B}^{(R_a\cup R_b)}(0)$ up to a product of two unitaries, one acting on $R_a$ and one acting on $R_b$. As the interaction terms $\Pi^{(a,b)}(\delta)$ (the projectors onto the orthogonal complement of the image of $\hat{B}^{(R_a\cup R_b)}(0)$) are pairwise commuting,  as we showed in the proof of Lemma \ref{t:block_ham_gap}, so are $h_{a,b}(\delta)$. Thus the Hamiltonian $\sum_{\langle a,b\rangle}h_{a,b}(\delta)$ is commuting and gapped, with a system-size independent gap of $\Delta=1$. 

Now consider the terms 
\begin{equation}
    h_{a,b}'(\delta)=\left(\Lambda_a(\delta)\otimes\Lambda_b(\delta)\right) h_{a,b}(\delta)\left(\Lambda_a(\delta)\otimes\Lambda_b(\delta)\right)\,,
\end{equation}
and the Hamiltonian $\sum_{\langle a,b\rangle} h_{a,b}'(\delta)$. We see that $h_{a,b}(\delta)$ and $h_{a,b}'(\delta)$ have the properties required of $h_{ij}$ and $h_{ij}'$ in Corollary \eqref{t:commuting_gap}. We also have that the kernel of $h_{a,b}(\delta)$ coincides with the kernel of the projector $\Pi^{(a,b)}(\delta)$ (i.e., the image of $\hat{B}^{(a,b)}(\delta)$). Thus for each edge $(a,b)$ we have that $h'_{a,b}(\delta)\le \eta(\delta)\Pi^{(a,b)}(\delta)$, where $\eta(\delta)$ is the largest eigenvalue of $h'_{a,b}(\delta)$. Therefore, $H^B(\delta)=\sum_{\langle a,b \rangle}\Pi^{(a,b)}(\delta)$ is gapped for any $\delta$ where $\sum_{\langle a,b\rangle} h_{a,b}'(\delta)$ is gapped. From Corollary \ref{t:commuting_gap} we thus conclude that $H^B(\delta)$ is gapped when the largest eigenvalue $\mu(\delta)$ of $\Lambda_a(\delta)$ satisfies the inequality Eq. \eqref{e:gap_inequality}. Note that from the definition of $\Lambda_a(\delta)$ that $\mu(\delta)\iv$ can be computed as the smallest non-zero eigenvalue of $(\hat{B}^{(R_a)}(\delta))^\dag\hat{B}^{(R_a)}(\delta)$, which can be expressed as a simple tensor contraction.

\bibliographystyle{apsrev4-1}

\end{document}